\pdfoutput=1
\documentclass[sigconf,authorversion]{acmart}

\usepackage{algorithm}
\usepackage[noend]{algpseudocode}      %%% this is for the algorithms
\algrenewcomment[1]{\hfill \(\triangleright\) \emph{\small #1}} % to italicize text in comments
\algnewcommand{\InlineIf}[2]{ % single line if-then
  \State \algorithmicif\ #1\ \algorithmicthen\ #2}
\algnewcommand{\InlineIfElse}[3]{ % single line if-then-else
  \State \algorithmicif\ #1\ \algorithmicthen\ #2\ \algorithmicelse\ #3}
\algnewcommand{\InlineFor}[2]{\algorithmicfor\ #1\ \algorithmicdo\ #2} % single line for loop
\usepackage{array}
\usepackage[capitalise]{cleveref}
\DeclareMathOperator*{\argmin}{arg\,min}

\newcommand{\bigO}[1]{O(#1)} % big O for complexity
\newcommand{\softO}[1]{\mathchoice{\tilde{O}\left(#1\right)}{O\tilde{~}(#1)}{O\tilde{~}(#1)}{O\tilde{~}(#1)}} % soft O for complexity
\newcommand{\expmm}{\omega} % exponent for the cost of matrix multiplication
\newcommand{\maxsz}{M} % maximal size of GB during the algorithm
\newcommand{\comp}{\mathcal{C}} % complexity of Pade
\newcommand{\algoname}[1]{{\normalfont\textsc{#1}}} % typesetting for algo names
\newcommand{\assign}{\leftarrow} % assign value to variable
\newcommand{\ZZp}{\mathbb{Z}_{> 0}} % positive integers
\newcommand{\field}{\mathbb{K}} % base field
\newcommand{\nvar}{r}  % number of variables X_1,...,X_r
\newcommand{\ring}{\mathcal{R}} % polynomial ring
\newcommand{\polMod}[1]{\ring^{#1}} % polynomial module K[X]^m
\newcommand{\matRing}[2]{\field^{#1 \times #2}} % scalar matrix space
\newcommand{\pmatRing}[2]{\ring^{#1 \times #2}} % polynomial matrix space, 2 opt args
\newcommand{\ideal}[1][I]{\mathcal{#1}} % some ideal (default I)
\newcommand{\module}[1][M]{\mathcal{#1}} % some module (default M)
\newcommand{\nodule}{\module[N]} % another module
\newcommand{\monom}[1]{\mathrm{Mon}(#1)} % monomials of a module M -> Mon(M)
\newcommand{\genBy}[1]{\langle #1 \rangle} % notation for ideal/module generated by polynomials
\newcommand{\col}[1]{\boldsymbol{#1}} % for a column of a matrix
\newcommand{\mat}[1]{\boldsymbol{#1}} % for a matrix
\newcommand{\xcol}{\col{X}} % column of the variables
\newcommand{\trsp}[1]{#1^\mathsf{T}} %transpose
\newcommand{\evec}[1]{\boldsymbol{e}_{#1}} % coordinate vector with 1 at index #1
\newcommand{\ident}[1]{\mat{I}_{#1}} % identity matrix of size #1 x #1
\newcommand{\matz}{\mat{0}}  % zero matrix
\newcommand{\ord}{\preccurlyeq} % monomial order
\newcommand{\ordneq}{\prec} % monomial order (strict)
\newcommand{\sord}[1]{\ord_{#1}} % Schreyer order wrt #1
\newcommand{\ordLex}{\ord_{\mathrm{lex}} } % lexicographic order on the polynomial ring
\newcommand{\ordTOPLex}{\ord^{\mathrm{top}}_{\mathrm{lex}} } % term over position lexicographic order on the polynomial module
\newcommand{\gb}{\mat{P}} % Grobner basis
\newcommand{\egb}{\mat{E}} % elementary Groebner basis
\newcommand{\gbb}{\mat{Q}} % second generic name for Grobner basis
\newcommand{\lm}[2]{\mathrm{lm}_{#1}(#2)} % leading monomial of polynomial #2 w.r.t. order #1
\newcommand{\gbdim}{k} % Grobner Basis "dim": number of elements in Grobner basis
\newcommand{\gbbdim}{\ell} % second generic notation for Grobner Basis "dim": number of elements in Grobner basis
\newcommand{\p}{\boldsymbol{p}} % generic name for a module polynomial / used for relation vectors
\newcommand{\pp}{p} % one of the components of a \p
\newcommand{\q}{\boldsymbol{q}} % second generic name for a module polynomial
\newcommand{\qq}{q} % one of the components of a \p
\newcommand{\f}{\boldsymbol{f}} % input polynomials giving the "equation" for syzygies
\newcommand{\ff}{f} % one of the components of a \f
\newcommand{\g}{\boldsymbol{g}} % second input polynomials giving the "equation" for syzygies (typically the residual)
\newcommand{\h}{\boldsymbol{h}} % generic name (2) for a module polynomial / module polynomial in second basis
\newcommand{\mmu}{\boldsymbol{\mu}} % monomial mu in R^m
\newcommand{\F}{\mat{F}} % tuple of input module polynomials (equivalent to a matrix: same notation)
\renewcommand{\G}{\mat{G}} % tuple of input module polynomials (equivalent to a matrix: same notation)
\renewcommand{\H}{\mat{H}} % tuple of input module polynomials (equivalent to a matrix: same notation)
\newcommand{\LM}{\mat{L}} % tuple of input (leading) monomials (equivalent to a matrix: same notation)
\newcommand{\LMinp}{\mat{K}} % input list of monomials of algorithm (to make difference with output)
\newcommand{\pt}{\mat{\alpha}} % pivot
\newcommand{\rdim}{m} % Relation DIM: m input elements, the relation module is in K[X]^m;
\newcommand{\edim}{n} % Element/Equation DIM: the input elements are in K[X]^n (or some quotient K[X]^n / M)
\newcommand{\syzmod}[2]{\mathrm{Syz}_{{#1}}(#2)} % relation module
\newcommand{\vsdim}{D} % dimension of the "input quotient" as a vector space
\newcommand{\lf}{\varphi} % linear form
\newcommand{\val}{\upsilon} % value (when evaluating a linear form)
\newcommand{\pivot}{\pi} % pivot
\newcommand{\prc}{d} % "precision" for Pade approximation (order of truncation)
\copyrightyear{2020} 
\acmYear{2020} 
\setcopyright{acmcopyright}\acmConference[ISSAC '20]{International Symposium on Symbolic and Algebraic Computation}{July 20--23, 2020}{Athens, Greece}
\acmBooktitle{International Symposium on Symbolic and Algebraic Computation (ISSAC '20), July 20--23, 2020, Athens, Greece}
\acmPrice{15.00}
\acmDOI{10.1145/3373207.3404059}
\acmISBN{978-1-4503-7100-1/20/07}
\begin{document}

\title{A Divide-and-conquer Algorithm for Computing Gr\"obner Bases of Syzygies in Finite Dimension}

\author{Simone Naldi}
\affiliation{%
  \institution{{\normalsize Univ. Limoges, CNRS, XLIM, UMR\,7252}}
  \city{F-87000 Limoges}
  \state{France}
}
%\email{simone.naldi@unilim.fr}

\author{Vincent Neiger}
\affiliation{%
  \institution{{\normalsize Univ. Limoges, CNRS, XLIM, UMR\,7252}}
  \city{F-87000 Limoges}
  \state{France}
}
%\email{vincent.neiger@unilim.fr}

\begin{abstract}
  Let \(\f_1, \ldots, \f_\rdim\) be elements in a quotient \(\polMod{\edim} /
  \nodule\) which has finite dimension as a \(\field\)-vector space, where
  \(\ring = \field[X_1,\ldots,X_\nvar]\) and \(\nodule\) is an
  \(\ring\)-submodule of \(\polMod{\edim}\). We address the problem of
  computing a Gr\"obner basis of the module of syzygies of \((\f_{\!1}, \ldots,
  \f_{\!\rdim})\), that is, of vectors \((\pp_1, \ldots, \pp_\rdim) \in
  \polMod{\rdim}\) such that \(\pp_1\f_{\!1}+\cdots+\pp_\rdim\f_{\!\rdim} = 0\).

  An iterative algorithm for this problem was given by Marinari, M\"oller, and
  Mora (1993) using a dual representation of \(\polMod{\edim}/\nodule\) as the
  kernel of a collection of linear functionals. Following this viewpoint, we
  design a divide-and-conquer algorithm, which can be interpreted as a
  generalization to several variables of Beckermann and Labahn's recursive
  approach for matrix Pad\'e and rational interpolation problems. To highlight
  the interest of this method, we focus on the specific case of bivariate
  Pad\'e approximation and show that it improves upon the best known complexity
  bounds.
\end{abstract}


\maketitle

\section{Introduction}
\label{sec:intro}

\paragraph{Context.}

Hereafter, \(\ring = \field[X_1,\ldots,X_\nvar]\) is the ring of
\(\nvar\)-variate polynomials over a field \(\field\). Given an
\(\ring\)-submodule \(\nodule\subset \polMod{\edim}\) such that
\(\polMod{\edim}/\nodule\) has finite dimension \(\vsdim\) as a
\(\field\)-vector space, as well as a matrix \(\F \in
\pmatRing{\rdim}{\edim}\) with rows \(\f_{\!1},\ldots,\f_{\!\rdim} \in
\polMod{\edim}\), this paper studies the computation of a Gr\"obner basis of
the module of syzygies
\begin{align*}
  \syzmod{\nodule}{\F} = \{\p = (\pp_i)_{1\le i\le \rdim} \in \ring^\rdim \mid \p \F = \textstyle\sum_{1\le i\le\rdim}\pp_i \f_{\!i} \in \nodule\},
\end{align*}
where \(\p\) is seen as a \(1 \times \rdim\) row vector. Note that
\(\polMod{\rdim}/\syzmod{\nodule}{\F}\) also has finite dimension, at most
\(\vsdim\), as a \(\field\)-vector space.

Following a path of work pioneered by Marinari, M\"oller and Mora
\cite{MaMoMo93,AlonsoMarinariMora2003,Mora09}, we focus on a specific situation
where \(\nodule\) is described using duality. That is, \(\nodule\) is known
through \(\vsdim\) linear functionals \(\lf_j : \polMod{\edim} \to \field\)
such that \(\nodule = \cap_{1 \le j \le \vsdim} \ker(\lf_j)\). In this context,
it is customary to make an assumption equivalent to the following: \(\nodule_i
= \cap_{1 \le j \le i} \ker(\lf_i)\) is an \(\ring\)-module, for \(1 \le i \le
\vsdim\); see e.g.~\cite[Algo.\,2]{MaMoMo93} \cite[Eqn.\,(4.1)]{Fitzpatrick97}
\cite[Eqn.\,(5)]{OKeeFit02} for such assumptions and related algorithms.
Namely, this assumption allows one to design iterative algorithms which compute
bases of \(\syzmod{\nodule_i}{\F}\) iteratively for increasing \(i\), until
reaching \(i=\vsdim\) and obtaining the sought basis of
\(\syzmod{\nodule}{\F}\). An efficient such iterative procedure is given in
\cite{MaMoMo93}, specifically in Algorithm 2 (variant in Section 9 therein);
note that it is written for \(\rdim=\edim=1\) and \(\F=[1]\), in which case
\(\syzmod{\nodule_i}{\F} = \nodule_i\), but directly extends to the case
\(\rdim\ge 1\) and \(\F \in \pmatRing{\rdim}{\edim}\).

\paragraph{Ideal of points and Pad\'e approximation.}

One particular case of interest is when \(\nodule\) is the vanishing ideal of a
given set of points: \(\edim=1\), and \(\nodule\) is the ideal of all
polynomials in \(\ring\) which vanish at distinct points
\(\pt_1,\ldots,\pt_\vsdim \in \field^\nvar\). Here, one takes the linear
functionals for evaluation: \(\lf_j : f \in \ring \mapsto f(\pt_j) \in
\field\). The question is, given the points, \(\rdim\) polynomials as \(\F \in
\pmatRing{\rdim}{1}\), and a monomial order \(\ord\), to compute a
\(\ord\)-Gr\"obner basis of the set of vectors \(\p\) such that \(\p\F\)
vanishes at all the points. When \(\rdim=1\) and \(\F=[1]\), this means
computing a \(\ord\)-Gr\"obner basis of the ideal of the points, as studied
in \cite{MolBuc82,MaMoMo93}.

Another case is that of (multivariate) Pad\'e approximation and its extensions,
as studied in \cite{FitFly92,Fitzpatrick97,OKeeFit02,farr2006computing}, as
well as in \cite{BerthomieuFaugere2018} in the context of the computation of
multidimensional linear recurrence relations. The basic setting is for
\(\edim=1\), with \(\nodule\) an ideal of the form
\(\genBy{X_1^{\prc_1},\ldots,X_\nvar^{\prc_\nvar}}\), and \(\F =
[\begin{smallmatrix} f \\ -1\end{smallmatrix}]\) for some given \(f \in
\ring\). Then, elements of \(\syzmod{\nodule}{\F}\) are vectors \((q,p) \in
\ring^2\) such that \(f = p/q \bmod X_1^{\prc_1},\ldots,X_\nvar^{\prc_\nvar}\).
Here, the \(\vsdim=\prc_1\cdots\prc_\nvar\) linear functionals correspond to
the coefficients of multidegree less than \((\prc_1,\ldots,\prc_\nvar)\); note
that not all orderings of these functionals satisfy the assumption above.

For these two situations, as well as some extensions of them, the fastest known
algorithms rely on linear algebra and have a cost bound of
\(\bigO{\rdim\vsdim^2 + \nvar \vsdim^3}\) operations in \(\field\)
\cite{MaMoMo93,Fitzpatrick97}; this was recently improved in
\cite[Thm.\,2.13]{Neiger16} and \cite{NeigerSchost19} to
\(\bigO{\rdim\vsdim^{\expmm-1} + \nvar \vsdim^\expmm\log(\vsdim)}\) where
\(\expmm < 2.38\) is the exponent of matrix multiplication
\cite{CopWin90,LeGall14}.

Based on work in \cite{CerliencoMureddu1995,FelszeghyRathRonyai2006}, in the
specific case of an ideal of points \(\nodule\) and the lexicographic order,
Ceria and Mora gave a combinatorial algorithm to compute the
\(\ordLex\)-monomial basis of \(\ring/\nodule\), the Cerlienco-Mureddu
correspondence, and squarefree separators for the points using \(\bigO{\nvar
\vsdim^2 \log(\vsdim)}\) operations \cite{ceria2018combinatorics}.

\paragraph{The univariate case.}

This problem has received attention in the case of a single variable
(\(\nvar=1\)) notably thanks to the numerous applications of matrix rational
interpolation and Hermite-Pad\'e approximation, which are the two situations
described above. Iterative algorithms were first given for Pad\'e approximation
in \cite{Wynn1960,Geddes73} and then for Hermite-Pad\'e approximation in
\cite{Beckermann92,BarBul92,BecLab97}; the latter can be seen as univariate
analogues of \cite[Algo.\,2]{MaMoMo93} and \cite[Algo.\,4.7]{Fitzpatrick97}.

A breakthrough divide and conquer approach was designed by Beckermann and
Labahn in \cite[Algo.\,SPHPS]{BecLab94}, allowing one to take advantage of
univariate polynomial matrix multiplication while previous iterative algorithms
only relied on naive linear algebra operations. This led to a line of work
\cite{GiJeVi03,Storjohann06,ZhoLab12,JeNeScVi16,JeNeScVi17} which consistently
improved the incorporation of fast linear algebra and fast polynomial
multiplication in this divide and conquer framework, culminating in cost bounds
for rational interpolation and Hermite-Pad\'e approximation which are close
asymptotically to the size of the problem (if \(\expmm=2\), these cost bounds
are quasi-linear in the size of the input). To the best of our knowledge, no
similar divide and conquer technique has been developed in multivariate
settings prior to this work. 

\paragraph{Contribution.}

We propose a divide and conquer algorithm for the problem of computing a
\(\ord\)-Gr\"obner basis of \(\syzmod{\nodule}{\F}\) in the multivariate case.
This is based on the iterative algorithm \cite[Algo.\,2]{MaMoMo93}, observing
that each step of the iteration can be interpreted as a left multiplication by
a matrix which has a specific shape, which we call elementary Gr\"obner basis
(see \cref{sec:basecase}). The new algorithm reorganizes these matrix products
through a divide and conquer strategy, and thus groups several products by
elementary Gr\"obner bases into a single multivariate polynomial matrix
multiplication.

Thus, both the existing iterative and the new divide and conquer approaches
compute the same elementary Gr\"obner bases, but unlike the former, our
algorithm does not explicitly compute Gr\"obner bases for all intermediate
syzygy modules \(\syzmod{\nodule_i}{\F}\). By computing less, we expect to
achieve better computational complexity. To illustrate this, we specialize our
approach to multivariate matrix Pad\'e approximation and derive complexity
bounds for this case; we obtain the next result, which is a particular case of
\cref{prop:bivariate_pade}.

\begin{theorem}
  \label{thm:bivariate_pade}
  For \(\ring = \field[X,Y]\), let \(\ff_1,\ldots,\ff_\rdim \in \ring\), and
  let \(\ord\) be a monomial order on \(\ring\). Then one can compute a minimal
  \(\ord\)-Gr\"obner basis of the module of Hermite-Pad\'e approximants
  \[
    \{(\pp_1,\ldots,\pp_\rdim) \in \polMod{\rdim} \mid \pp_1 \ff_1 + \cdots + \pp_\rdim \ff_\rdim = 0 \bmod \genBy{X^d,Y^d}\}
  \]
  using \(\softO{\rdim^\expmm d^{\expmm+2}}\) operations in \(\field\), where
  \(\softO{\cdot}\) means that polylogarithmic factors are omitted.
\end{theorem}

In this case the vector space dimension is \(\vsdim=d^2\). Thus, as noted above
and to the best of our knowledge, the fastest previously known algorithm for
this task has a cost of \(\softO{\rdim d^{2(\expmm-1)} + d^{2\expmm}}\)
operations in \(\field\) and does not exploit fast polynomial multiplication.

\paragraph{Perspectives.}

The base case of our divide and conquer algorithm concerns the case \(\nodule =
\ker(\lf)\) of a single linear functional, detailed in \cref{sec:basecase}; we
thus work in a vector space \(\polMod{\edim}/\nodule\) of dimension \(1\). A
natural perspective is to improve the efficiency of our algorithm thanks to a
better exploitation of fast linear algebra by grouping several base cases
together; using fast linear algebra to accelerate the base case was a key
strategy in obtaining efficient univariate algorithms
\cite{GiJeVi03,JeNeScVi17}. In the context of Pad\'e approximation, where one
can introduce the variables one after another, one could also try to
incorporate known algorithms for the univariate case.

One reason why these improvements are not straightforward to do in the
multivariate case is that there is no direct generalization of a property at
the core of the correctness of univariate algorithms. This property (see
\cite[Lem.\,2.4]{JeaNeiVil2020}) states that if \(\gb_1\) is a
\(\ord_1\)-Gr\"obner basis of \(\nodule_1 \supset \nodule\) and \(\gb_2\) is a
\(\ord_2\)-Gr\"obner basis of \(\syzmod{\nodule}{\gb_1}\), then \(\gb_2\gb_1\)
is a \(\ord_1\)-Gr\"obner basis of \(\nodule\), provided that the order
\(\ord_2\) is well chosen (a Schreyer order for \(\gb_1\) and \(\ord_1\), see
\cref{sec:preliminaries:schreyer}). We give a counterexample to such a property
in \cref{eg:product_gb_counterex}. It remains open to find a similar general
property that would help to design algorithms based on matrix multiplication in
the multivariate case.

Another difficulty arises in analyzing the complexity of our divide and conquer
scheme in contexts where the number of elements in the sought Gr\"obner basis
is not well controlled, such as rational interpolation. Indeed, this number
corresponds to the size of the matrices used in the algorithm, and therefore is
directly related to the cost of the matrix multiplication. In fact, the
worst-case number of elements depends on the monomial order and is often
pessimistic compared to what is observed in a generic situation. Thus, future
work involves investigating complexity bounds for generic input and for
interesting particular cases other than Pad\'e approximation.

\section{Preliminaries}
\label{sec:preliminaries}

\subsection{Notation}
\label{sec:preliminaries:notation}

Here and hereafter, the coordinate vector with \(1\) at index \(i\) is denoted
by \(\evec{i}\); its dimension is inferred from the context.  A monomial in
\(\polMod{\rdim}\) is an element of the form \(\nu\evec{i}\) for some \(1 \le i
\le \rdim\) and some monomial \(\nu\) in \(\ring\); \(i\) is called the support
of \(\nu\evec{i}\).  We denote by \(\monom{\polMod{\rdim}}\) the set of all
monomials in \(\polMod{\rdim}\).  A term is a monomial multiplied by a nonzero
constant from \(\field\). The elements of \(\polMod{\rdim}\) are
\(\field\)-linear combinations of elements of \(\monom{\polMod{\rdim}}\) and
are called polynomials.

Elements in \(\ring\) are written in regular font (e.g.~monomials \(\mu\) and
\(\nu\) and polynomials \(\ff\) and \(\pp\)), while elements in
\(\polMod{\rdim}\) are boldfaced (e.g.~monomials \(\mmu\) and
\(\boldsymbol{\nu}\) and polynomials \(\f\) and \(\p\)). Vectors or (ordered)
lists of polynomials in \(\polMod{\rdim}\) are seen as matrices, written
in boldfaced capital letters; precisely, \((\p_1,\ldots,\p_\gbdim) \in
(\polMod{\rdim})^\gbdim\) is seen as a matrix \(\mat{P} \in
\pmatRing{\gbdim}{\rdim}\) whose \(i\)th row is \(\p_i\). In particular, in
what follows the default orientation is to see an element of \(\polMod{\rdim}\)
as a \emph{row} vector in \(\pmatRing{1}{\rdim}\).

For the sake of completeness, we recall below in
\cref{sec:preliminaries:monomial_order,sec:preliminaries:gb,sec:preliminaries:schreyer}
some classical definitions from commutative algebra concerning submodules of
\(\polMod{\rdim}\); we assume familiarity with the corresponding notions
concerning ideals of \(\ring\). For a more detailed introduction the reader may
refer to \cite{CoLiOSh05,CoLiOSh07,Eisenbud95}. 

\subsection{Monomial orders for modules}
\label{sec:preliminaries:monomial_order}

A monomial order on \(\polMod{\rdim}\) is a total order \(\ord\) on
\(\monom{\polMod{\rdim}}\) such that, for \(\nu \in \monom{\ring}\) and
\(\mmu_1,\mmu_2 \in \monom{\polMod{\rdim}}\) with \(\mmu_1 \ord \mmu_2\), one
has \(\mmu_1 \ord \nu \mmu_1 \ord \nu \mmu_2\); hereafter \(\mmu_1 \ordneq
\mmu_2\) means that \(\mmu_1\ord\mmu_2\) and \(\mmu_1 \neq \mmu_2\). For \(\p
\in \polMod{\rdim}\), its \(\ord\)-leading monomial is denoted by
\(\lm{\ord}{\p}\) and is the largest of its monomials with respect to the order
\(\ord\) (we take the convention \(\lm{\ord}{\matz} = \matz\) for \(\matz \in
\polMod{\rdim}\) the zero element). We extend this notation to collections of
polynomials \(\mathcal{P} \subset \polMod{\rdim}\) with
\(\lm{\ord}{\mathcal{P}} = \{\lm{\ord}{\p} : \p \in \mathcal{P}\}\), and to
matrices \(\gb \in \pmatRing{\gbdim}{\rdim}\) with \(\lm{\ord}{\gb}\) the
\(\gbdim \times \rdim\) matrix whose $i$th row is the \(\ord\)-leading monomial
of the \(i\)th row of \(\gb\).

\begin{example}
  \label{ex:ord_toplex}
  The usual lexicographic comparison is a monomial order on \(\field[X,Y]\):
  \(X^a Y^b \ordLex X^{a'} Y^{b'}\) if and only if \(a < a'\) or (\(a=a'\) and
  \(b < b'\)). It can be used to define a monomial order on \(\field[X,Y]^2\),
  called the term-over-position lexicographic order: for \(\mu,\nu\) in
  \(\monom{\field[X,Y]}\) and \(i,j\) in \(\{1,2\}\), \(\mu \evec{i}
  \ordTOPLex \nu \evec{j}\) if and only if \(\mu \ordLex \nu\) or (\(\mu =
  \nu\) and \(i<j\)).
\end{example}
We refer to \cite[Sec.\,1.\S2 and 5.\S2]{CoLiOSh05} for other classical
monomial orders, such as the degree reverse lexicographical order on
\(\ring\), and the construction of term-over-position and position-over-term
orders on \(\polMod{\rdim}\) from monomial orders on \(\ring\).

A monomial order \(\ord\) on \(\polMod{\rdim}\) induces a monomial order
\(\ord_i\) on \(\ring\) for each \(1 \le i \le \rdim\), by restricting to the
\(i\)th coordinate: for \(\nu_1,\nu_2 \in \monom{\ring}\), \(\nu_1 \ord_i
\nu_2\) if and only if \(\nu_1 \evec{i} \ord \nu_2 \evec{i}\). In particular,
\(\lm{\ord}{\qq \p}\) is a multiple of \(\lm{\ord}{\p}\) for \(\qq \in \ring\)
and \(\p \in \polMod{\rdim}\):
\begin{lemma}
  \label{lem:leading_modules}
  Let \(\bar{\imath}\) be the support of \(\lm{\ord}{\p}\). Then \(\lm{\ord}{\qq \p} =
  \lm{\ord_{\bar{\imath}}}{\qq} \lm{\ord}{\p}\).
\end{lemma}
\begin{proof}
  Write \(\qq = \sum_{\ell} \nu_\ell\) and \({\p} = \sum_{i,j} \mu_{ij}
  \evec{i}\) for terms \(\mu_{ij},\nu_\ell\) in \(\ring\). Then \({\qq\p} =
  \sum_{\ell, i, j} \nu_\ell\mu_{ij} \evec{i}\), i.e.~the terms of \({\qq\p}\)
  are all those of the form \(\nu_\ell\mu_{ij} \evec{i}\).  Now let
  \(\bar\ell\) and \(\bar{\jmath}\) be such that
  \(\lm{\ord_{\bar{\imath}}}{\qq} = \nu_{\bar\ell}\) and \(\lm{\ord}{\p} =
  \mu_{\bar{\imath}\bar{\jmath}} \evec{\bar{\imath}}\). Then \(\nu_{\ell}
  \ordneq_{\bar{\imath}} \nu_{\bar\ell}\) for all \(\ell\neq\bar\ell\), which
  implies that \(\nu_{\ell}\mu_{\bar{\imath}\bar{\jmath}}
  \ordneq_{\bar{\imath}} \nu_{\bar\ell}\mu_{\bar{\imath}\bar{\jmath}}\) and
  thus, by definition of \(\ord_{\bar{\imath}}\), that
  \(\nu_{\ell}\mu_{\bar{\imath}\bar{\jmath}}\evec{\bar{\imath}} \ordneq
  \nu_{\bar\ell}\mu_{\bar{\imath}\bar{\jmath}}\evec{\bar{\imath}}\). On the
  other hand, \(\mu_{ij} \evec{i} \ordneq \mu_{\bar{\imath}\bar{\jmath}}
  \evec{\bar{\imath}}\) holds for all \((i,j) \neq
  (\bar{\imath},\bar{\jmath})\), hence \(\nu_\ell\mu_{ij} \evec{i} \ordneq
  \nu_\ell\mu_{\bar{\imath}\bar{\jmath}} \evec{\bar{\imath}}\). Therefore we
  obtain \(\nu_\ell\mu_{ij} \evec{i} \ord
  \nu_{\bar\ell}\mu_{\bar{\imath}\bar{\jmath}} \evec{\bar{\imath}}\) for all
  \((i,j,\ell)\), with equality only if \((i,j,\ell) =
  (\bar\imath,\bar\jmath,\bar\ell)\). This proves that \(\lm{\ord}{\qq \p} =
  \nu_{\bar\ell} \mu_{\bar{\imath}\bar{\jmath}} \evec{\bar{\imath}} =
  \lm{\ord_{\bar{\imath}}}{\qq}\lm{\ord}{\p}\).
\end{proof}

\subsection{Gr\"obner bases}
\label{sec:preliminaries:gb}

As a consequence of Hilbert's Basis Theorem, any \(\ring\)-submodule of
\(\polMod{\rdim}\) is finitely generated \cite[Prop.\,1.4]{Eisenbud95}. For a
(possibly infinite) collection of polynomials \(\mathcal{P} \subset
\polMod{\rdim}\), we denote by \(\genBy{\mathcal{P}}\) the \(\ring\)-submodule
of \(\polMod{\rdim}\) generated by the elements of \(\mathcal{P}\). Similarly,
for a matrix \(\gb\) in \(\pmatRing{\gbdim}{\rdim}\), \(\genBy{\gb}\) stands
for the \(\ring\)-submodule of \(\polMod{\rdim}\) generated by its rows, that
is, \(\genBy{\gb} = \{\q \gb \mid \q \in \polMod{\gbdim}\}\).

For a given submodule \(\module \subset \polMod{\rdim}\), the \(\ord\)-leading
module of \(\module\) is the module \(\genBy{\lm{\ord}{\module}}\) generated by
the leading monomials of the elements of \(\module\). Then, a matrix \(\gb\) in
\(\pmatRing{\gbdim}{\rdim}\) whose rows are in \(\module\) is said to be a
\(\ord\)-Gr\"obner basis of \(\module\) if
\[
  \genBy{\lm{\ord}{\module}} = \genBy{\lm{\ord}{\gb}} .
\]
%%% Rephrasing this definition: for \(\genBy{\lm{\ord}{\module}} =
%%% \genBy{\nu_1\evec{i_1}, \ldots, \nu_\gbdim\evec{i_\gbdim}}\) and \(\p_1,
%%% \ldots, \p_\gbdim \in \module\) satisfying \(\lm{\ord}{\p_j} = \nu_j
%%% \evec{i_j}\) for all \(j\in\{1,\ldots,\gbdim\}\), then \(\gb =
%%% (\p_1,\ldots,\p_\gbdim)\) is a \(\ord\)-Gr\"obner basis of \(\module\), see
%%% \cite[Ch.5.2, Ex.4]{CoLiOSh05}.
In this case we have \(\genBy{\gb} = \module\) (see \cite[Ch.5,
Prop.2.7]{CoLiOSh05}), hence we will often omit the reference to the module
\(\module\) and just say that \(\gb\) is a \(\ord\)-Gr\"obner basis.

A \(\ord\)-Gr\"obner basis \(\gb\), whose rows are \((\p_1,\ldots,\p_\gbdim)\),
is said to be minimal if \(\lm{\ord}{\p_i}\) is not divisible by
\(\lm{\ord}{\p_j}\), for any \(j\neq i\). It is said to be reduced if it is
minimal and, for all \(1 \le i \le \gbdim\), \(\lm{\ord}{\p_i}\) is monic and
none of the terms of \(\p_i\) is divisible by any of \(\{\lm{\ord}{\p_j} \mid
j\neq i\}\). Given a monomial order \(\ord\) and an \(\ring\)-submodule
\(\module\subset \polMod{\rdim}\), there is a reduced \(\ord\)-Gr\"obner basis
of \(\module\) and it is unique (up to permutation of its elements)
\cite[Sec.\,15.2]{Eisenbud95}.

\begin{example}
  \label{ex:module_codim_one}
  The syzygy module
  \[
    \module = \{(\pp_1,\pp_2) \in \field[X,Y]^2 \mid \pp_1 - \pp_2 \in \genBy{X,Y}\}
    = \syzmod{\genBy{X,Y}}{[\begin{smallmatrix} 1 \\ -1 \end{smallmatrix}]}
  \]
  is generated by \((X\evec{1}, Y\evec{1}, \evec{1} + \evec{2})\), that is, by the
  rows of
  \[
    \gb =
    \begin{bmatrix}
      X & 0 \\
      Y & 0 \\
      1 & 1
    \end{bmatrix} \in \field[X,Y]^{3\times 2}.
  \]
  Furthermore, \(\gb\) is the reduced \(\ordTOPLex\)-Gr\"obner basis of \(\module\).
\end{example}

\subsection{Schreyer orders}
\label{sec:preliminaries:schreyer}

In the context of the computation of bases of syzygies it is generally
beneficial to use a specific construction of monomial orders, as first
highlighted by Schreyer \cite{Schreyer1980,Janet1920} (see also
\cite[Th.\,15.10]{Eisenbud95} and \cite{BerkeschSchreyer2015}).

In the univariate case, the notion of shifted degree plays the same role as
Schreyer orders and is ubiquitous in the computation of bases of modules of
syzygies \cite{GiJeVi03,ZhoLab12,JeNeScVi16}; an equivalent notion of defects
was also used earlier for M-Pad\'e and Hermite-Pad\'e approximation algorithms
\cite{Beckermann92,BecLab94}. Specifically, this provides a monomial order on
\(\polMod{\gbdim}\) constructed from a monomial order \(\ord\) on
\(\polMod{\rdim}\) and from the leading monomials of a \(\ord\)-Gr\"obner basis
in \(\polMod{\rdim}\) of cardinality \(\gbdim\).

\begin{definition}
    \label{dfn:schreyer_order}
  Let \(\ord\) be a monomial order on \(\polMod{\rdim}\), and let \(\LM =
  (\mmu_1, \ldots, \mmu_\gbdim)\) be a list of monomials of \(\polMod{\rdim}\).
  A \emph{Schreyer order for \(\ord\) and \(\LM\)} is any monomial order on
  \(\polMod{\gbdim}\), denoted by \(\sord{\LM}\), such that for \(\nu_1
  \evec{i}, \nu_2 \evec{j} \in \monom{\polMod{\gbdim}}\), if \({\nu_1\mmu_i}
  \ordneq {\nu_2 \mmu_j}\) then \(\nu_1\evec{i} \sord{\LM} \nu_2\evec{j}\).
\end{definition}
\noindent As noted above, this notion is often used with \(\LM =
\lm{\ord}{\gb}\) for a list of polynomials \(\gb \in \pmatRing{\gbdim}{\rdim}\),
which is typically a \(\ord\)-Gr\"obner basis.
%For simplicity, in this case we use the shorthand notation \(\sord{\gb}\) for
%\(\sord{\lm{\ord}{\gb}}\) and we say that \(\sord{\gb}\) is a Schreyer order
%for \(\ord\) and \(\gb\).

Remark that \cref{dfn:schreyer_order} uses a strict inequality, and implies that if
\(\nu_1\evec{i} \sord{\LM} \nu_2\evec{j}\), then \({\nu_1 \mmu_i} \ordneq \nu_2
\mmu_j\) or \(\nu_1 \mmu_i = \nu_2 \mmu_j\). In particular, for
\(\nu_1=\nu_2=1\) and assuming \(\mmu_i \neq \mmu_j\) for all \(i\neq j\) (for
instance, if \(\LM=\lm{\ord}{\gb}\) for a minimal \(\ord\)-Gr\"obner basis
\(\gb\)), then \(\evec{i} \sord{\LM} \evec{j}\) if and only if \(\mmu_i \ordneq
\mmu_j\).

Furthermore, for every \(\ord\) and \(\LM\), a corresponding Schreyer order
exists and can be constructed explicitly: for example, \(\nu_1\evec{i}
\sord{\LM} \nu_2\evec{j}\) if and only if
\[
\nu_1 \mmu_i \ordneq \nu_2 \mmu_j
\text{ or } (\nu_1 \mmu_i = \nu_2 \mmu_j \text{ and } i < j).
\]
This specific Schreyer order is the one used in the algorithms in this paper,
where we write
\[
  \sord{\LM} { } \assign \algoname{SchreyerOrder}(\ord,\LM)
\]
to mean that the algorithm constructs it from \(\ord\) and \(\LM\).

\section{Base case of the divide and conquer scheme}
\label{sec:basecase}

In this section we present the base case of our main algorithm. It constructs
Gr\"obner bases for syzygies modulo the kernel of a single linear functional,
which we call elementary Gr\"obner bases and describe in
\cref{sec:basecase:egb}. Further in \cref{sec:codim_one:gb} we state properties
that are useful to prove the correctness of the base case algorithm given in
\cref{sec:basecase:algo}. Precisely, this correctness is written having in mind
the design of an algorithm handling several functionals iteratively by
repeating this basic procedure and multiplying the elementary bases together.

\subsection{Elementary Gr\"obner basis}
\label{sec:basecase:egb}

If \(\ideal \subset \ring\) is an ideal such that \(\ring/\ideal\) has
dimension \(1\) as a \(\field\)-vector space, then \(\ideal\) is maximal: it is
of the form \(\genBy{X_1-\alpha_1,\ldots,X_\nvar-\alpha_\nvar}\) for some point
\((\alpha_1,\ldots,\alpha_\nvar) \in \field^\nvar\), which directly yields the
reduced Gr\"obner basis of \(\ideal\), for any monomial order. In this paper,
we will make use of a similar property for submodules of \(\polMod{\rdim}\);
such submodules have Gr\"obner bases of the form
\begin{equation}
  \label{eqn:egb}
  \egb =
  \begin{bmatrix}
    \ident{\pivot-1} & \mat{\lambda}_1 \\
                   & \xcol - \mat{\alpha} \\
                   & \mat{\lambda}_2      & \ident{\rdim-\pivot}
  \end{bmatrix}
  \in \pmatRing{(\rdim+\nvar-1)}{\rdim},
\end{equation}
for the vector of variables \(\xcol = \trsp{[X_1 \; \cdots \; X_\nvar]}\) and
vectors of values \(\mat{\alpha} = \trsp{[\alpha_1 \; \cdots \; \alpha_\nvar]}
\in \matRing{\nvar}{1}\), \(\mat{\lambda}_1 = \trsp{[\lambda_1 \; \cdots \;
\lambda_{\pivot-1}]} \in \matRing{(\pivot-1)}{1}\), and \(\mat{\lambda}_2 =
\trsp{[\lambda_{\pivot+1} \; \cdots \; \lambda_{\rdim}]} \in
\matRing{(\rdim-\pivot)}{1}\). In what follows, such matrices are called
elementary Gr\"obner bases.

%% \begin{align*}
%%   \gb = \{ & \evec{1} + \lambda_1 \evec{\pivot}, \ldots, \evec{\pivot-1} + \lambda_{\pivot-1} \evec{\pivot}, \\
%%            & (X_1-\alpha_1) \evec{\pivot}, \ldots, (X_\nvar-\alpha_\nvar) \evec{\pivot}, \\
%%            & \lambda_{\pivot+1} \evec{\pivot} + \evec{\pivot+1}, \ldots, \lambda_{\rdim} \evec{\pivot} + \evec{\rdim}\},
%% \end{align*}
%% for some
%% \(\lambda_1,\ldots,\lambda_{\pivot-1},\lambda_{\pivot+1},\ldots,\lambda_\rdim \in
%% \field\) and \(\alpha_1,\ldots,\alpha_\nvar \in \field\), with \(\lambda_i =
%% 0\) if \(\evec{i} \ordneq \evec{\pivot}\) for \(i \neq \pivot\).

\begin{theorem}
  \label{thm:codim_one}
  Let \(\module\) be an \(\ring\)-submodule of \(\polMod{\rdim}\) such that
  \(\polMod{\rdim}/\module\) has dimension \(1\) as a \(\field\)-vector space,
  then for any monomial order \(\ord\) on \(\polMod{\rdim}\), the reduced
  \(\ord\)-Gr\"obner basis \(\egb\) of \(\module\) is as in
  \cref{eqn:egb} with \(\lambda_i = 0\) if \(\evec{i} \ordneq
  \evec{\pivot}\) for all \(i \neq \pivot\).
  Conversely, any matrix \(\egb\) as in \cref{eqn:egb} defines a
  submodule \(\module = \genBy{\egb}\) such that \(\polMod{\rdim}/\module\) has
  dimension \(1\) as a \(\field\)-vector space, and \(\egb\) is a reduced
  \(\ord\)-Gr\"obner basis for any monomial order \(\ord\) such that
  \(\lambda_i = 0\) if \(\evec{i} \ordneq \evec{\pivot}\) for all \(i \neq
  \pivot\).
\end{theorem}
\begin{proof}
  By \cite[Thm.\,15.3]{Eisenbud95}, a basis of \(\polMod{\rdim}/\module\) as a
  \(\field\)-vector space is given by the monomials not in
  \(\lm{\ord}{\module}\); since the dimension of \(\polMod{\rdim}/\module\) as
  a \(\field\)-vector space is \(1\), there exists a unique monomial which is
  not in \(\lm{\ord}{\module}\). Thus there is a unique \(\pivot \in
  \{1,\ldots,\rdim\}\) such that
  \begin{equation}
  \label{eqn:lm_codim_one}
    \lm{\ord}{\egb} = (\evec{1}, \ldots, \evec{\pivot-1}, X_1 \evec{\pivot}, \ldots, X_\nvar \evec{\pivot}, \evec{\pivot+1}, \ldots, \evec{\rdim}).
  \end{equation}
  By definition of reduced Gr\"obner bases, the \(j\)th polynomial in \(\egb\)
  is the sum of the \(j\)th element of \(\lm{\ord}{\egb}\) and a constant
  multiple of \(\evec{\pivot}\); hence \(\egb\) has the form in
  \cref{eqn:egb}. In addition, for \(i\neq \pivot\), the equality
  \(\lm{\ord}{\evec{i} + \lambda_i \evec{\pivot}} = \evec{i}\) implies that
  \(\lambda_i = 0\) whenever \(\evec{i} \ordneq \evec{\pivot}\).

  For the converse, let \({\ord}\) be such that \(\lambda_i = 0\) if \(\evec{i}
  \ordneq \evec{\pivot}\) for all \(i \neq \pivot\) (such an order exists since
  there are orders for which \(\evec{\pivot}\) is the smallest coordinate
  vector). Then \(\lm{\ord}{\egb}\) is as in \cref{eqn:lm_codim_one}; in
  particular, the monomials in \(\genBy{\lm{\ord}{\egb}}\) are precisely
  \(\monom{\polMod{\rdim}}\setminus\{\evec{\pivot}\}\). It follows that either
  \(\evec{\pivot} \in \lm{\ord}{\module}\) and \(\genBy{\lm{\ord}{\module}} =
  \polMod\rdim\), or \(\evec{\pivot} \not\in \lm{\ord}{\module}\) and
  \(\genBy{\lm{\ord}{\module}}=\genBy{\lm{\ord}{\egb}}\). In the second case
  \(\egb\) is a reduced \(\ord\)-Gr\"obner-basis and \(\polMod{\rdim}/\module\)
  has dimension 1 by \cite[Thm.\,15.3]{Eisenbud95}. To conclude the proof, we
  show that \(\evec{\pivot} \in \lm{\ord}{\module}\) cannot occur; by
  contradiction, suppose there exists \(\q \in \module\) such that
  \(\lm{\ord}{\q} = \evec{\pivot}\). Since the rows of \(\egb\) generate
  \(\module\), we can write
  \begin{align*}
    \q & = (\qq_1,\ldots,\qq_{\pivot-1},\pp_1,\ldots,\pp_\nvar,\qq_{\pivot+1},\ldots,\qq_\rdim) \,\egb \\
       & = \left(\qq_1,\ldots,\qq_{\pivot-1},\sum_{i\neq\pivot}\qq_i\lambda_i+\sum_{j=1}^\nvar(X_j-\alpha_j) \pp_j,\qq_{\pivot+1},\ldots,\qq_\rdim\right).
  \end{align*}
  For \(i \neq \pivot\) such that \(\evec{\pivot} \ordneq \evec{i}\), any
  nonzero term of \(\qq_i \evec{i}\) would appear in \(\q\) and be greater than
  \(\evec{\pivot}\), hence \(\qq_i = 0\). Moreover, for \(i \neq \pivot\) such
  that \(\evec{i} \ordneq \evec{\pivot}\) we have \(\lambda_i = 0\). Thus,
  considering the \(\pivot\)th component of \(\q\) yields the equality
  \[
    1 = \sum_{i\neq\pivot}\qq_i\lambda_i+\sum_{j=1}^\nvar(X_j-\alpha_j) \pp_j = \sum_{j=1}^\nvar(X_j-\alpha_j)\pp_j
  \]
  which is a contradiction since \(1 \not\in
  \genBy{X_1-\alpha_1,\ldots,X_\nvar-\alpha_\nvar}\).
\end{proof}

Remark that in the module case (\(\rdim \geq 2\)) the reduced
\(\ord\)-Gr\"obner basis depends on the order \(\ord\), more precisely on how
the \(\evec{i}\)'s are ordered by \(\ord\). For instance, the matrix in
\cref{ex:module_codim_one} is a reduced \(\ord\)-Gr\"obner basis for every
order such that \(\evec{1} \ord \evec{2}\), whereas for orders such that
\(\evec{2} \ord \evec{1}\) the reduced \(\ord\)-Gr\"obner basis of the same
module is
\[
  \egb =
  \begin{bmatrix}
    1 & 1 \\
    0 & X \\
    0 & Y
  \end{bmatrix} \in \field[X,Y]^{3\times 2}.
\]

\subsection{Multiplying by elementary Gr\"obner bases}
\label{sec:codim_one:gb}

Let \(\ord\) be a monomial order on \(\polMod{\rdim}\) and let \(\gb = (\p_1,
\ldots, \p_\gbdim) \in \pmatRing{k}{\rdim}\) be a \(\ord\)-Gr\"obner basis. In
this section, we show conditions on an elementary Gr\"obner basis \(\egb\) to
ensure that \(\egb \gb\) is a \(\ord\)-Gr\"obner basis.

We write \(\LM = (\mmu_1,\ldots,\mmu_\gbdim)\) for \(\lm{\ord}{\gb}\), that is,
\(\mmu_i = \lm{\ord}{\p_i}\) for \(1 \le i \le \gbdim\). Let \(\sord{\LM}\) be
a Schreyer order for \(\ord\) and \(\gb\), and consider a reduced
\(\sord{\LM}\)-Gr\"obner basis \(\egb \in \pmatRing{(\gbdim+\nvar-1)}{\gbdim}\)
which has the form in \cref{eqn:egb}; thus
\begin{align*}
  \egb\gb = ( & \p_{1} + \lambda_1 \p_{\pivot}, \ldots, \p_{\pivot-1} + \lambda_{\pivot-1} \p_{\pivot}, \\
  & (X_1-\alpha_1) \p_{\pivot}, \ldots, (X_\nvar-\alpha_\nvar) \p_{\pivot}, \\
  & \lambda_{\pivot+1} \p_{\pivot} + \p_{\pivot+1}, \ldots, \lambda_{\gbdim} \p_{\pivot} + \p_{\gbdim})
\end{align*}
which is in \(\pmatRing{(\gbdim+\nvar-1)}{\rdim}\). We will show that, under
suitable assumptions, \(\egb\gb\) is a \(\ord\)-Gr\"obner basis; the next
lemmas use the above notation. We start by describing the leading terms of
\(\egb\gb\).

\begin{lemma}
  \label{lem:leading}
  If \(\mmu_i \neq \mmu_\pivot\) for all \(i \neq \pivot\), then
  \begin{align*}
    \lm{\ord}{\egb\gb} & = \lm{\sord{\LM}}{\egb} \, \LM \\
                       & = (\mmu_1, \ldots, \mmu_{\pivot-1},
                           X_1 \mmu_\pivot, \ldots, X_\nvar \mmu_\pivot,
                           \mmu_{\pivot+1}, \ldots, \mmu_{\gbdim}).
 \end{align*}
\end{lemma}
\begin{proof}
  First, \(\lm{\ord}{(X_j-\alpha_j) \p_{\pivot}} = X_j \mmu_\pivot\) for \(1\le
  j\le \nvar\). Next we claim that \(\lm{\ord}{\p_i + \lambda_i \p_{\pivot}} =
  \mmu_i\) for all \(i \neq \pivot\). If \(\lambda_i=0\), the identity is
  obvious. If \(\lambda_i \neq 0\), then \(\evec{\pivot} \sord{\LM} \evec{i}\)
  (see \cref{sec:basecase:egb}), and from the definition of a Schreyer
  order and the assumption \(\mmu_\pivot \neq \mmu_i\), we deduce \(\mmu_\pivot
  \ordneq \mmu_i\) and hence \(\lm{\ord}{\p_i + \lambda_i \p_{\pivot}} =
  \mmu_i\).
\end{proof}

Next, we characterize the fact that \(\egb\gb\) generates a submodule which
differs from the one generated by \(\gb\).

\begin{lemma}
  \label{lem:equivalence}
  If \(\mmu_i \neq \mmu_\pivot\) for all \(i \neq \pivot\), then
  \[
    \genBy{\egb\gb} \neq \genBy{\gb}
    \;\Leftrightarrow\;
    \p_{\pivot} \not\in \genBy{\egb\gb}
    \;\Leftrightarrow\;
    \mmu_\pivot \not\in \genBy{\lm{\ord}{\genBy{\egb\gb}}}.
  \]
\end{lemma}
\begin{proof}
  First, remark that \(\genBy{\egb\gb} = \genBy{\gb} \Rightarrow {\p_{\pivot}}
  \in \genBy{\egb\gb} \Rightarrow \mmu_\pivot \in
  \genBy{\lm{\ord}{\genBy{\egb\gb}}}\) is obvious; thus, to conclude the proof
  it remains to show that \(\genBy{\egb\gb} = \genBy{\gb} \Leftarrow
  \mmu_\pivot \in \genBy{\lm{\ord}{\genBy{\egb\gb}}}\). Suppose that
  \(\mmu_\pivot \in \genBy{\lm{\ord}{\genBy{\egb\gb}}}\). Then,
  since \(\mmu_i \in \genBy{\lm{\ord}{\genBy{\egb\gb}}}\) for all
  \(i \neq \pivot\) by \cref{lem:leading}, we have \(\lm{\ord}{{\gb}} \subset
  \genBy{\lm{\ord}{\genBy{\egb\gb}}}\), hence \(\genBy{\lm{\ord}{{\gb}}}
  \subset \genBy{\lm{\ord}{\genBy{\egb\gb}}}\). Furthermore, recall that
  \(\genBy{\lm{\ord}{\gb}} = \genBy{\lm{\ord}{\genBy{\gb}}}\) since \(\gb\) is
  a \(\ord\)-Gr\"obner basis, and that \(\genBy{\lm{\ord}{\genBy{\egb\gb}}}
  \subset \genBy{\lm{\ord}{\genBy{\gb}}}\) since \(\genBy{\egb\gb} \subset
  \genBy{\gb}\): we obtain \(\genBy{\lm{\ord}{\genBy{\gb}}} =
  \genBy{\lm{\ord}{\genBy{\egb\gb}}}\). Then, \cite[Lemma 15.5]{Eisenbud95}
  shows that \(\genBy{\egb\gb} = \genBy{\gb}\).
\end{proof}

For example, if \(\gb\) is a \emph{minimal} \(\ord\)-Gr\"obner basis, then the
assumption in the previous lemma is satisfied. \Cref{eg:product_gb_counterex} below
exhibits a case where \(\gb\) is a minimal \(\ord\)-Gr\"obner basis and
\(\p_\pivot\) does belong to \(\genBy{\egb\gb}\). In that case,
\(\genBy{\egb\gb} = \genBy{\gb}\) and \(\egb\gb\) is not a Gr\"obner basis
since \(\mmu_\pivot\) is in $\genBy{\lm{\ord}{\genBy{\egb\gb}}}$ but not in
$\genBy{\lm{\ord}{\egb\gb}}$.

\begin{lemma}
  \label{lem:mul_by_one}
  If \(\mmu_i \neq \mmu_\pivot\) for all \(i \neq \pivot\) and
  \(\genBy{\egb\gb} \neq \genBy{\gb}\), then \(\egb\gb\) is a
  \(\ord\)-Gr\"obner basis.
\end{lemma}
\begin{proof}
  Suppose by contradiction that \(\egb\gb\) is not a \(\ord\)-Gr\"obner basis.
  Then there exists a nonzero \(\h \in \genBy{\egb\gb}\) such that
  \(\lm{\ord}{\h} \not\in \genBy{\lm{\ord}{\egb\gb}}\), that is, by
  \cref{lem:leading}, \(\lm{\ord}{\h}\) is not divisible by any of the elements
  \(\mmu_i\) for \(i\neq\pivot\) and \(X_j \mmu_\pivot\) for \(1\le j\le
  \nvar\). On the other hand, \(\lm{\ord}{\h}\) is in
  \(\genBy{\lm{\ord}{\genBy{\egb\gb}}}\) and therefore in
  \(\genBy{\lm{\ord}{\gb}}\), hence \(\lm{\ord}{\h}\) is divisible by at least
  one \(\mmu_i\), \(1 \le i \le \gbdim\). These divisibility constraints lead
  to \(\lm{\ord}{\h} = \mmu_\pivot\), which implies \(\mmu_\pivot \in
  \genBy{\lm{\ord}{\genBy{\egb\gb}}}\). From \cref{lem:equivalence} one deduces
  \(\genBy{\egb\gb} = \genBy{\gb}\), which is absurd.
\end{proof}

\begin{corollary}
  \label{cor:mul_by_one}
  Assume that \(\genBy{\egb\gb} \neq \genBy{\gb}\) and that \(\gb\) is a
  minimal \(\ord\)-Gr\"obner basis. Let \(j_1 < \cdots < j_\ell\) be the
  indices \(j \in \{1,\ldots,\nvar\}\) such that \(X_j \mmu_\pivot \not\in
  \genBy{\mmu_i, i\neq \pivot}\). Then, the submatrix
  %%\begin{align*}
  %%  \gbb
  %%  =
  %%  ( & \evec{1} + \lambda_1 \evec{\pivot}, \ldots, \evec{\pivot-1} + \lambda_{\pivot-1} \evec{\pivot}, \\
  %%    & (X_{j_1}-\alpha_{j_1}) \evec{\pivot}, \ldots, (X_{j_\ell}-\alpha_{j_\ell}) \evec{\pivot}, \\
  %%    & \lambda_{\pivot+1} \evec{\pivot} + \evec{\pivot+1}, \ldots, \lambda_{\gbdim} \evec{\pivot} + \evec{\gbdim}),
  %%\end{align*}
  \begin{equation}
    \label{eqn:egb_pruned}
    \gbb =
    \begin{bmatrix}
      \ident{\pivot-1} & \mat{\lambda}_1 \\
                       & X_{j_1}-\alpha_{j_1} \\
                       & \vdots \\
                       & X_{j_\ell}-\alpha_{j_\ell} \\
                       & \mat{\lambda}_2      & \ident{\rdim-\pivot}
    \end{bmatrix}
    \in \pmatRing{(\gbdim+\ell-1)}{\gbdim}
  \end{equation}
  of \(\egb\) is such that \(\gbb\gb\) is a minimal \(\ord\)-Gr\"obner basis of
  \(\genBy{\egb\gb}\).
\end{corollary}
\begin{proof}
  Since \(\gb\) is minimal, \(\mmu_i \neq \mmu_\pivot\) for all \(i \neq
  \pivot\); then \cref{lem:mul_by_one} ensures that \(\egb\gb\) is a
  \(\ord\)-Gr\"obner basis and \cref{lem:leading} gives
  \[
    \lm{\ord}{\gbb\gb} = (\mmu_1, \ldots, \mmu_{\pivot-1},
                           X_{j_1} \mmu_\pivot, \ldots, X_{j_\ell} \mmu_\pivot,
                           \mmu_{\pivot+1}, \ldots, \mmu_{\gbdim}).
  \]
  By construction of \(j_1, \ldots, j_\ell\), one has
  \(\genBy{\lm{\ord}{\gbb\gb}} = \genBy{\lm{\ord}{{\egb\gb}}}\), which implies
  \begin{align*}
    \genBy{\lm{\ord}{\genBy{\egb\gb}}} = \genBy{\lm{\ord}{{\egb\gb}}} & = \genBy{\lm{\ord}{{\gbb\gb}}} \\
                                       & \subset \genBy{\lm{\ord}{\genBy{\gbb\gb}}} \subset \genBy{\lm{\ord}{\genBy{\egb\gb}}}.
  \end{align*}
  Hence \(\genBy{\lm{\ord}{{\gbb\gb}}} = \genBy{\lm{\ord}{\genBy{\gbb\gb}}}\),
  and \(\gbb\gb\) is a minimal \(\ord\)-Gr\"obner basis. We conclude using
  \cite[Lem.\,15.5]{Eisenbud95}, which shows that \(\genBy{\gbb\gb} \subset
  \genBy{\egb\gb}\) and \(\genBy{\lm{\ord}{\genBy{\gbb\gb}}} =
  \genBy{\lm{\ord}{\genBy{\egb\gb}}}\) imply
  \(\genBy{\gbb\gb}=\genBy{\egb\gb}\).
\end{proof}

\begin{example}
  \label{eg:product_gb_counterex}
  Consider the case \(\ring=\field[X,Y]\) and \(\rdim=1\). Let \(\gb =
  [\begin{smallmatrix} X \\ Y+1 \end{smallmatrix}] \in \pmatRing{2}{1}\), which
  is the reduced \(\ord_1\)-Gr\"obner basis of \(\genBy{X, Y+1}\) for any
  monomial order \(\ord_1\) on \(\monom{\ring}\). Let also \(\egb \in
  \pmatRing{3}{2}\) whose rows are \((X\evec{1}, Y\evec{1}, \evec{2})\);
  according to \cref{thm:codim_one}, \(\egb\) is a reduced \(\ord_2\)-Gr\"obner
  basis for any monomial order \(\ord_2\) on \(\monom{\polMod{2}}\). Now, the
  product \(\egb\gb \in \pmatRing{3}{1}\) has entries \(X^2\), \(XY\), and
  \(Y+1\). Thus, \(\genBy{\lm{\ord_3}{\egb\gb}} = \genBy{X^2,XY,Y} =
  \genBy{X^2,Y}\) for any monomial order \(\ord_3\) on \(\monom{\ring}\). On
  the other hand, \(\genBy{\egb\gb}\) contains \(X = X(Y+1) - XY\), hence
  \(\genBy{\lm{\ord_3}{\egb\gb}} \neq \genBy{\lm{\ord_3}{\genBy{\egb\gb}}}\),
  which means that \(\egb\gb\) is not a \(\ord_3\)-Gr\"obner basis.
\end{example}

\subsection{Algorithm}
\label{sec:basecase:algo}

We now describe Algorithm~\algoname{Syzygy\_BaseCase}, which will serve as the
base case of the divide and conquer scheme.

\begin{algorithm}
  \caption{\algoname{Syzygy\_BaseCase}$(\lf,\G,\ord,\LM)$}
  \label{algo:basecase}
  \begin{algorithmic}[1]
    \Require{
      \Statex \textbullet~ a linear functional $\lf : \ring^\edim \to \field$,
      \Statex \textbullet~ a matrix \(\G\) in \(\pmatRing{\gbdim}{\edim}\) with rows \(\g_1,\ldots,\g_\gbdim \in \polMod{\edim}\),
      \Statex \textbullet~ a monomial order \(\ord\) on \(\ring^\rdim\), %% for some \(\rdim\in\ZZp\),
      \Statex \textbullet~ a list \(\LMinp = (\mmu_1, \ldots, \mmu_\gbdim)\) of elements of \(\monom{\ring^\rdim}\).
    }
    \Ensure{
      \Statex \textbullet~ a matrix \(\gbb\) in \(\pmatRing{(\gbdim+\ell-1)}{\gbdim}\) for some \(\ell \in \{0,\ldots,\nvar\}\),
      \Statex \textbullet~ a list \(\LM\) of \(\gbdim+\ell-1\) elements of \(\monom{\ring^\gbdim}\).
    }

    \State \((\val_1,\ldots,\val_\gbdim) \assign (\lf(\g_1),\ldots,\lf(\g_\gbdim)) \in \field^\gbdim\)
    \InlineIf{\((\val_1,\ldots,\val_\gbdim) = (0,\ldots,0)\)}{\Return \((\ident{\gbdim}, \LMinp)\)}
    \label{algo:basecase:step:trivial}
    \State \( \sord{\LMinp} { } \assign \algoname{SchreyerOrder}(\ord,\LMinp)\)
    \State \(\pivot \assign \argmin_{\sord{\LMinp}} \{\evec{i} \mid 1 \le i \le \gbdim, \val_i \neq 0\}\)
    \Comment{the index \(i\) such that \(\val_i \neq 0\) which minimizes \(\evec{i}\) with respect to \(\sord{\LMinp}\)}
    \label{algo:basecase:step:pivot}
    \State \(\{j_1 < \cdots < j_\ell\} \assign \{j \in \{1,\ldots,\nvar\} \mid X_j \mmu_\pivot \not\in \genBy{\mmu_i, i \neq \pivot} \}\)
    \State \(\alpha_{j_s} \assign \lf({X_{j_s}\g_\pivot}) / \val_\pivot\) for \(1 \le s \le \ell\)
          \label{algo:basecase:step:alphas}
    \State \(\lambda_i \assign -\val_i/\val_\pivot\) for \(1 \le i < \pivot\) and \(\pivot < i \le \gbdim\)
          \label{algo:basecase:step:lambdas}
    \State \(\gbb \assign\) matrix in \(\pmatRing{(\gbdim+\ell-1)}{\gbdim}\) as
    in \cref{eqn:egb_pruned}
    %%\State \(\!\!\!\begin{array}{lll}
    %%\gbb & \assign & ( \evec{1} + \lambda_1 \evec{\pivot}, \ldots, \evec{\pivot-1} + \lambda_{\pivot-1} \evec{\pivot}, \\
    %%          & & (X_{j_1}-\alpha_{j_1}) \evec{\pivot}, \ldots, (X_{j_\ell}-\alpha_{j_\ell}) \evec{\pivot}, \\
    %%                & & \lambda_{\pivot+1} \evec{\pivot} + \evec{\pivot+1}, \ldots, \lambda_{\gbdim} \evec{\pivot} + \evec{\gbdim} )
    %%\end{array}\)
    \label{algo:basecase:step:Q}
    \State $\LM \assign (\mmu_{1},\ldots,\mmu_{\pivot-1},X_{j_1}\mmu_{\pivot},\ldots,X_{j_\ell}\mmu_{\pivot},\mmu_{\pivot+1},\ldots,\mmu_{\gbdim})$
    \label{algo:basecase:step:LM}
    \State \Return $(\gbb, \LM)$
  \end{algorithmic}
\end{algorithm}

\begin{theorem}
  \label{thm:algo:basecase}
  Let \(\nodule \subset \polMod{\edim}\) be an \(\ring\)-submodule, let \(\F
  \in \pmatRing{\rdim}{\edim}\), and let \(\gb \in \pmatRing{\gbdim}{\rdim}\)
  be a minimal \(\ord\)-Gr\"obner basis of \(\syzmod{\nodule}{\F}\) for some
  monomial order \(\ord\) on \(\ring^\rdim\). Assume that the input of
  \cref{algo:basecase} is such that \(\ker(\lf) \cap \nodule\) is an
  \(\ring\)-module, \(\G = \gb \F\), and \(\lm{\ord}{\gb} =
  (\mmu_1,\ldots,\mmu_\gbdim)\). Then \cref{algo:basecase} returns \((\gbb,
  \LM)\) such that \(\gbb\gb\) is a minimal \(\ord\)-Gr\"obner basis of
  \(\syzmod{\ker(\lf)\cap\nodule}{\F}\) and \(\LM=\lm{\ord}{\gbb\gb}\).
\end{theorem}
\begin{proof}
  If \((\lf(\g_1),\ldots,\lf(\g_\gbdim)) = (0,\ldots,0)\), then
  \cref{algo:basecase} stops at \cref{algo:basecase:step:trivial} and returns
  \(\gbb=\ident{\gbdim}\) and \(\LMinp\). Thus \(\gbb\gb=\gb\), hence by
  assumption \(\LM = \LMinp = \lm{\ord}{\gb} = \lm{\ord}{\gbb\gb}\), and
  \(\gbb\gb\) is a minimal \(\ord\)-Gr\"obner basis of
  \(\syzmod{\nodule}{\F}\); besides, the identity \(\syzmod{\nodule}{\F} =
  \syzmod{\ker(\lf)\cap\nodule}{\F}\) is easily deduced from
  \((\lf(\g_1),\ldots,\lf(\g_\gbdim)) = (0,\ldots,0)\).

  In the rest of the proof, assume \((\lf(\g_1),\ldots,\lf(\g_\gbdim)) \neq
  (0,\ldots,0)\). Define \(\egb \in \pmatRing{(\gbdim+\nvar-1)}{\gbdim}\) as in
  \cref{eqn:egb} with \(\pivot\) and \(\lambda_i\) as in \cref{algo:basecase}
  and \(\alpha_j = \lf(X_j\g_\pivot)/\val_\pivot\) for \(1 \le j \le \nvar\);
  in particular, \(\gbb\) computed at \cref{algo:basecase:step:Q} is formed by
  a subset of the rows of \(\egb\). 

  First, \(\egb\) is a \(\sord{\LMinp}\)-Gr\"obner basis according to
  \cref{thm:codim_one}, since by definition of \(\pivot\) and \(\lambda_i\) one
  gets the implications \(\evec{i} \sord{\LMinp} \evec{\pivot} \Rightarrow
  \val_i = 0 \Rightarrow \lambda_i = 0\), for \(i \neq \pivot\).

  Next, we claim that \(\genBy{\egb} = \syzmod{\ker(\lf)\cap\nodule}{\G}\).
  Indeed, the rows of \(\gb\F\) are in \(\nodule\), and thus so are those of
  \(\egb\G=\egb\gb\F\). Moreover, by choice of \(\pivot\) and \(\lambda_i\) the
  rows of \(\egb\G\) are in \(\ker(\lf)\), since for \(i \neq \pivot\) one has
  \(\lf((\p_i + \lambda_i\p_\pivot)\F) = \lf(\g_i + \lambda_i\g_\pivot) =
  \val_i + \lambda_i\val_\pivot = 0\) and for \(1 \le j \le \nvar\) one has
  \(\lf((X_{j}-\alpha_{j}) \p_{\pivot}\F) = \lf((X_{j}-\alpha_{j}) \g_{\pivot})
  = \lf(X_{j}\g_{\pivot})-\alpha_{j}\val_{\pivot} = 0\).  Therefore the rows of
  \(\egb\G\) are in \(\ker(\lf) \cap \nodule\), that is, \(\genBy{\egb} \subset
  \syzmod{\ker(\lf)\cap\nodule}{\G}\). To prove the reverse inclusion, recall
  from \cref{thm:codim_one} that \(\genBy{\egb}\) has codimension \(1\) in
  \(\polMod{\gbdim}\) and hence \(\syzmod{\ker(\lf)\cap\nodule}{\G}\) is either
  \(\genBy{\egb}\) or \(\polMod{\gbdim}\). Since 
  \[
  0 \neq \val_\pivot = \lf(\g_\pivot) = \lf(\p_\pivot\F) = \lf(\evec{\pivot}\gb\F) = \lf(\evec{\pivot}\G)
  \]
  one has that \(\evec{\pivot} \not\in \syzmod{\ker(\lf)\cap\nodule}{\G}\),
  hence \(\syzmod{\ker(\lf)\cap\nodule}{\G} = \genBy{\egb}\).

  It follows that \(\genBy{\egb\gb} = \syzmod{\ker(\lf) \cap \nodule}{\F}\).
  Indeed, the rows of \(\egb\gb\F\) are in \(\ker(\lf) \cap \nodule\) as noted
  above, and thus \(\genBy{\egb\gb} \subset \syzmod{\ker(\lf) \cap
  \nodule}{\F}\). Now let \(\p \in \syzmod{\ker(\lf) \cap \nodule}{\F}\); thus
  in particular \(\p \in \syzmod{\nodule}{\F}\), and \(\p = \q\gb\) for some
  \(\q \in \polMod{\gbdim}\). Then \(\p\F = \q\gb\F = \q\G \in \ker(\lf) \cap
  \nodule\), hence \(\q \in \syzmod{\ker(\lf)\cap\nodule}{\G} = \genBy{\egb}\),
  and therefore \(\p \in \genBy{\egb\gb}\).

  Now, \(\lf(\p_\pivot\F) \neq 0\) implies \(\p_\pivot \not\in
  \syzmod{\ker(\lf)\cap \nodule}{\F} = \genBy{\egb\gb}\). Thus
  \cref{lem:equivalence} ensures \(\genBy{\egb\gb} \neq \genBy{\gb}\), and
  finally \cref{cor:mul_by_one} states that \({\gbb\gb}\) is a minimal
  \(\ord\)-Gr\"obner basis of \(\genBy{\egb\gb} =
  \syzmod{\ker(\lf)\cap\nodule}{\F}\). Besides \cref{lem:leading} yields
  \(\lm{\ord}{\gbb\gb} = \lm{\sord{\LMinp}}{\gbb}\LMinp = \LM\).
\end{proof}

\section{Divide and conquer algorithm}
\label{sec:dac}

Repeating the basic procedure described in \cref{sec:basecase:algo}
iteratively, we obtain an algorithm for syzygy basis computation when
\(\nodule\) is an intersection of kernels of linear functionals with a specific
property (see \cref{hyp:kernel_lf_module}). This algorithm is similar to
\cite[Algo.\,2]{MaMoMo93} and \cite[Algo.\,3.2]{OKeeFit02}, apart from
differences in the input description. Here, the input consists of linear
functionals \(\lf_1, \ldots, \lf_\vsdim : \polMod{\edim} \to \field\), with the
assumption that
\begin{equation}
  \label{hyp:kernel_lf_module}
  \nodule_i = \displaystyle\cap_{1\le j\le i} \ker(\lf_j) \text{ is an } \ring\text{-module for } 1 \le i \le \vsdim.
\end{equation}
Then we consider the \(\ring\)-module \(\nodule = \nodule_\vsdim = \cap_{1\le
j\le \vsdim} \ker(\lf_j)\), which is such that \(\polMod{\edim}/\nodule\) has
dimension at most \(\vsdim\) as a \(\field\)-vector space. For \(\F\) in
\(\pmatRing{\rdim}{\edim}\), the following algorithm computes a minimal
\(\ord\)-Gr\"obner basis of the syzygy module \(\syzmod{\nodule}{\F}\). Note
that we do not specify the representation of \(\F\) since it may depend on the
specific functionals \(\lf_i\); typically, one considers \(\F\) to be known
modulo \(\nodule\), via the images of its rows by the functionals \(\lf_i\).

\begin{algorithm}
  \caption{\algoname{Syzygy\_Iter}$(\lf_1,\ldots,\lf_\vsdim,\F,\ord)$}
  \label{algo:iter}
  \begin{algorithmic}[1]
    \Require{
      \Statex \textbullet~ linear functionals $\lf_1,\ldots,\lf_\vsdim : \ring^\edim \to \field$ such that \cref{hyp:kernel_lf_module},
      \Statex \textbullet~ a matrix \(\F\) in \(\pmatRing{\rdim}{\edim}\),
      \Statex \textbullet~ a monomial order \(\ord\) on \(\ring^\rdim\).
     }
    \Ensure{
      \Statex \textbullet~ a minimal \(\ord\)-Gr\"obner basis \(\gb \in \pmatRing{\gbdim}{\rdim}\) of \(\syzmod{\nodule}{\F}\).
    }

    %% initialization
    \State \(\gb \assign \ident{\rdim} \in \pmatRing{\rdim}{\rdim}\); \(\G \assign \F\); \(\LM \assign (\evec{1},\ldots,\evec{\rdim}) = \lm{\ord}{\gb}\)
    \label{algo:iter:step:base}
    %% for loop: incorporate linear constraint \lf_i
    \For{$i=1,\ldots,\vsdim$} \label{algo:iter:for}
    \State \((\gbb,\LM) \assign \Call{Syzygy\_BaseCase}{\lf_i,\G,\ord,\LM}\)
    \State \(\gb \assign \gbb \gb\); \(\G \assign \gbb \G\)
    \EndFor
    \State\Return{$\gb$}
  \end{algorithmic}
\end{algorithm}

\begin{corollary}
  \label{cor:algo:iter}
  At the end of the \(i\)th iteration of \cref{algo:iter}, \(\gb\) is a minimal
  \(\ord\)-Gr\"obner basis of \(\syzmod{\nodule_i}{\F}\), and one has \(\G =
  \gb \F\) as well as \(\LM = \lm{\ord}{\gb}\). In particular, \cref{algo:iter}
  is correct.
\end{corollary}
\begin{proof}
  Note that at \cref{algo:iter:step:base} of \cref{algo:iter}, \(\gb = \ident{\rdim}\) is the
  reduced \(\ord\)-Gr\"obner basis of \(\polMod{\rdim} = \syzmod{\nodule_0}{\F}\)
  with \(\nodule_0 = \polMod{\edim}\), and both \(\G = \gb\F = \F\) and \(\LM =
  (\evec{1},\ldots,\evec{\rdim}) = \lm{\ord}{\gb}\) hold. We conclude that if
  \(\vsdim = 0\), \cref{algo:iter} is correct.

  The rest of the proof is by induction on \(\vsdim\). 
  We claim that the properties in the statement are preserved across the \(\vsdim\)
  iterations. Precisely, we assume that at the beginning of the \(i\)th iteration,
  \(\gb\) is a minimal \(\ord\)-Gr\"obner basis of \(\syzmod{\nodule_i}{\F}\),
  \(\G = \gb \F\), and \(\LM = \lm{\ord}{\gb}\).

  Since \(\nodule_{i+1} = \ker(\lf_{i+1}) \cap
  \nodule_i\) is an \(\ring\)-module, applying \cref{thm:algo:basecase} shows
  that \((\gbb,\LM)\) computed during the iteration are such that \(\LM =
  \lm{\ord}{\gbb\gb}\) and that \(\gbb\gb\) is a minimal \(\ord\)-Gr\"obner
  basis of \(\syzmod{\nodule_{i+1}}{\F}\).
\end{proof}

This allows us to deduce bounds on the size of a minimal \(\ord\)-Gr\"obner
basis of \(\syzmod{\nodule}{\F}\).

\begin{lemma}
  \label{lem:general:size_gb}
  Let \(\gb \in \pmatRing{\gbdim}{\rdim}\) be the output of \cref{algo:iter}.
  Then, \(\rdim \le \gbdim \le \rdim + (\nvar-1)\vsdim\), and thus the same
  holds for any minimal \(\ord\)-Gr\"obner basis of \(\syzmod{\nodule}{\F}\).
  Furthermore, at the end of the iteration \(i\) of \cref{algo:iter}, the
  basis \(\gbb\) has at most \(\gbdim + \vsdim - i\) elements.
\end{lemma}
\begin{proof}
  Remark that all minimal \(\ord\)-Gr\"obner bases of the same module have the
  same number of rows. Before the first iteration, the basis is
  \(\ident{\rdim}\) which has \(\rdim\) rows, and each iteration of the for
  loop adds \(\ell-1\) rows to the basis for some \(\ell\) in
  \(\{0,\ldots,\nvar\}\). Therefore \(\gbdim \le \rdim + (\nvar-1)\vsdim\), and
  the last claim follows from \(\ell-1\ge-1\). The lower bound \(\rdim \le
  \gbdim\) comes from the fact that \(\polMod{\rdim} / \syzmod{\nodule}{\F}\)
  has finite dimension as a \(\field\)-vector space.
\end{proof}

This iterative algorithm can be turned into a divide and conquer one
(\cref{algo:dac}), by reorganizing how the products are performed. It computes
a minimal \(\ord\)-Gr\"obner basis of \(\syzmod{\nodule}{\F}\), if one takes as
input \(\G=\F\) and \(\LMinp = (\evec{1},\ldots,\evec{\rdim})\).

\begin{algorithm}
  \caption{\algoname{Syzygy\_DaC}$(\lf_1,\ldots,\lf_\vsdim,\G,\ord,\LMinp)$}
  \label{algo:dac}
  \begin{algorithmic}[1]
    \Require{
      \Statex \textbullet~ linear functionals $\lf_1,\ldots,\lf_\vsdim : \ring^\edim \to \field$,
      \Statex \textbullet~ a matrix \(\G\) in \(\pmatRing{\gbdim}{\edim}\),
      \Statex \textbullet~ a monomial order \(\ord\) on \(\ring^\rdim\),
      \Statex \textbullet~ a list \(\LMinp = (\mmu_1, \ldots, \mmu_\gbdim)\) of elements of \(\monom{\ring^\rdim}\).
     }
    \Ensure{
      \Statex \textbullet~ a matrix \(\gbb\) in \(\pmatRing{\gbbdim}{\rdim}\) for some \(\gbbdim \ge 0\),
      \Statex \textbullet~ a list \(\LM\) of \(\gbbdim\) elements of \(\monom{\ring^\rdim}\).
    }
    \InlineIf{\(\vsdim=1\)}{\Return \Call{Syzygy\_BaseCase}{$\lf_i,\G,\ord,\LMinp$}}
    \State \((\gbb_1,\LM_1) \assign \algoname{Syzygy\_DaC}(\lf_1,\ldots,\lf_{\lfloor \vsdim/2 \rfloor},\G,\ord,\LMinp)\)
    \State \((\gbb_2,\LM_2) \assign \algoname{Syzygy\_DaC}(\lf_{\lfloor \vsdim/2 \rfloor+1},\ldots,\lf_\vsdim,\gbb_1\G,\ord,\LM_1)\)
    \State \Return \((\gbb_2 \gbb_1, \LM_2)\)
  \end{algorithmic}
\end{algorithm}

\begin{theorem}
  \label{thm:algo_general:dac}
  Let \(\nodule \subset \polMod{\edim}\) be an \(\ring\)-submodule, let \(\F
  \in \pmatRing{\rdim}{\edim}\), and let \(\gb \in \pmatRing{\gbdim}{\rdim}\)
  be a minimal \(\ord\)-Gr\"obner basis of \(\syzmod{\nodule}{\F}\) for some
  monomial order \(\ord\) on \(\ring^\rdim\). Assume that the input of
  \cref{algo:dac} is such that \(\G = \gb \F\), and \(\lm{\ord}{\gb} =
  (\mmu_1,\ldots,\mmu_\gbdim)\), and
  \begin{equation}
    \label{hyp:kernel_lfs_module}
    \nodule_i \cap \nodule \text{ is an } \ring\text{-module for } 1 \le i \le \vsdim,
  \end{equation}
  where \(\nodule_i = \cap_{1\le j\le i} \ker(\lf_j)\). Then \cref{algo:dac}
  outputs \((\gbb,\LM)\) such that \(\gbb\gb\) is a minimal \(\ord\)-Gr\"obner
  basis of \(\syzmod{\nodule_\vsdim\cap\nodule}{\F}\) and
  \(\LM=\lm{\ord}{\gbb\gb}\).
\end{theorem}
\begin{proof}
  If \(\vsdim = 1\) the output returned by \cref{algo:basecase} is correct, since by
  \cref{thm:algo:basecase}, \(\gbb\gb\) is a minimal \(\ord\)-Gr\"obner basis of                                                                                           \(\syzmod{\ker(\lf_1)\cap\nodule}{\F}\) and \(\LM=\lm{\ord}{\gbb\gb}\). We assume
  by induction hypothesis that \cref{algo:dac} returns the output foreseen by
  \cref{thm:algo_general:dac} when the number of input linear functionals is \(<\vsdim\),
  and when the assumptions of the theorem are satisfied.
  
  By such a hypothesis, since \(\G = \gb\F\) and \(\LMinp = \lm{\ord}{\gb}\), one
  deduces that \((\gbb_1, \LM_1)\) are such that \(\gbb_1\gb\) is a \(\ord\)-Gr\"obner 
  basis of \(\syzmod{\module[M]}{\F}\), with \(\module[M] = \nodule_{\lfloor \vsdim/2 \rfloor} \cap \nodule\),
  and \(\LM_1 = \lm{\ord}{\gbb_1\gb}\).

  Let \(\module[K]_i = \cap_{\lfloor \vsdim/2\rfloor + 1 \leq j \leq i} \ker(\lf_j)\),
  for each \(i=\lfloor D/2\rfloor+1, \ldots, \vsdim\). By hypothesis \(\module[K]_i \cap
  \module[M] = \nodule_i \cap \nodule\) is a module, for \(i=\lfloor \vsdim/2\rfloor+1\),
  \ldots, \(i=\vsdim\). Since \(\gbb_1\G = \gbb_1\gb\F\) and \(\gbb_1\gb\) is a \(\ord\)-Gr\"obner
  basis of \(\syzmod{\module[M]}{\F}\), and \(\LM_1 = \lm{\ord}{\gbb_1\gb}\),
  we can apply again the induction hypothesis, and conclude that \((\gbb_2,\LM_2)\) is
  such that \(\gbb_2\gbb_1\gb\) is a minimal \(\ord\)-Gr\"obner basis of
  \(\syzmod{\module[K]_\vsdim \cap \module[M]}{\F} = \syzmod{\nodule_\vsdim \cap \nodule}{\F}\),
  and \(\LM_2 = \lm{\ord}{\gbb_2\gbb_1\gb}\). We conclude that the global output \((\gbb_2\gbb_1,\LM_2)\)
  satisfies the claimed properties.
\end{proof}

\section{Multivariate Pad\'e approximation}
\label{sec:pade}

The algorithm in the previous section gives a general framework, which can be
refined when applied to a particular context. Here, we consider the context of
multivariate Pad\'e approximation, where
\begin{equation}
  \label{eqn:pade_module}
  \nodule = \genBy{X_1^{\prc_1}, \ldots, X_\nvar^{\prc_\nvar}} \times \cdots
  \times \genBy{X_1^{\prc_1}, \ldots, X_\nvar^{\prc_\nvar}} \subseteq
  \polMod{\edim},
\end{equation}
for some \(\prc_1,\ldots,\prc_\nvar \in \ZZp\). We begin with some remarks on
the degrees and sizes of Gr\"obner bases of syzygy modules
\(\syzmod{\nodule}{\F}\).

To express this context in the framework of \cref{sec:dac}, we take for the
\(\vsdim\) linear functionals \(\lf_i\) the dual basis of the canonical
monomial basis of \(\polMod{\edim}/\nodule\). Precisely, the linear functionals
are \(\lf_{\mu,j} : \polMod{\edim} \to \field\) for \(1 \le j \le \edim\) and
all monomials \(\mu \in \monom{\ring}\) with \(\deg_{X_i}(\mu) < \prc_i\) for
\(1 \le i \le \nvar\), defined as follows: for \(\f = (\ff_1,\ldots,\ff_\edim)
\in \polMod{\edim}\), \(\lf_{\mu,j}(\f)\) is the coefficient of the monomial
\(\mu\) in \(\ff_j\). These linear functionals can be ordered in several ways
to ensure that \cref{hyp:kernel_lf_module} is satisfied. Here we design our
algorithm by ordering the functionals \(\lf_{\mu,j}\) according to the
term-over-position lexicographic order on the monomials \(\mu \evec{j} \in
\monom{\polMod{\edim}}\).

\begin{example}
  Consider the case of \(\nvar=2\) variables \(X,Y\) with \(\prc_1 = 2\),
  \(\prc_2=4\), and \(\edim=2\). Then the functionals are
  \[
    \begin{array}{%
      >{\centering\arraybackslash$} p{0.44cm} <{$}
      >{\centering\arraybackslash$} p{0.44cm} <{$}
      >{\centering\arraybackslash$} p{0.6cm} <{$}
      >{\centering\arraybackslash$} p{0.65cm} <{$}
      >{\centering\arraybackslash$} p{0.7cm} <{$}
      >{\centering\arraybackslash$} p{0.7cm} <{$}
      >{\centering\arraybackslash$} p{0.7cm} <{$}
      >{\centering\arraybackslash$} p{0.7cm} <{$}}
      \lf_{1,1}, &
      \lf_{1,2}, &
      \lf_{Y,1}, &
      \lf_{Y,2}, &
      \lf_{Y^2,1}, &
      \lf_{Y^2,2}, &
      \lf_{Y^3,1}, &
      \lf_{Y^3,2}, \\
      \lf_{X,1}, &
      \lf_{X,2}, &
      \lf_{XY,1}, &
      \lf_{XY,2}, &
      \lf_{XY^2,1}, &
      \lf_{XY^2,2}, &
      \lf_{XY^3,1}, &
      \lf_{XY^3,2},
    \end{array}
  \]
  in this specific order.
\end{example}

\begin{lemma}
  \label{lem:pade:degree}
  Let \(\nodule\) be as in \cref{eqn:pade_module}, let \(\F \in
  \pmatRing{\rdim}{\edim}\), and let \(\ord\) be a monomial order on
  \(\polMod{\rdim}\). Then, for \(1 \le i \le \nvar\), each polynomial in the
  reduced \(\ord\)-Gr\"obner basis of \(\syzmod{\nodule}{\F}\) either has
  degree in \(X_i\) less than \(\prc_i\) or has the form \(X_i^{\prc_i}
  \evec{j}\) for some \(1 \le j \le \rdim\).
\end{lemma}
\begin{proof}
  Let \(\gb\) be the reduced \(\ord\)-Gr\"obner basis of
  \(\syzmod{\nodule}{\F}\) and let \(i \in \{1,\ldots,\nvar\}\). Since
  \(\polMod{\rdim}/\syzmod{\nodule}{\F}\) has finite dimension as a
  \(\field\)-vector space, for each \(j \in \{1,\ldots,\edim\}\) there is a
  polynomial in \(\gb\) whose \(\ord\)-leading monomial has the form \(X_i^\prc
  \evec{j}\) for some \(\prc \ge 0\). Since \(\gb\) is reduced, any other
  \((\pp_1,\ldots,\pp_\rdim)\) in \(\gb\) whose \(\ord\)-leading monomial has
  support \(j\) is such that \(\deg_{X_i}(\pp_j) < \prc \le \prc_i\); the last
  inequality follows from the fact that the monomial \(X_i^{\prc_i}\evec{j}\)
  is in \(\syzmod{\nodule}{\F}\) and thus is a multiple of \(X_i^\prc
  \evec{j}\). It follows that all polynomials in \(\gb\) whose \(\ord\)-leading
  monomial is not among \(\{X_i^{\prc_i} \evec{j}, 1 \le j \le \edim\}\) must
  have degree in \(X_i\) less than \(\prc_i\). On the other hand, any
  polynomial in \(\gb\) whose \(\ord\)-leading monomial is \(X_i^{\prc_i}
  \evec{j}\) for some \(j\) must be equal to this monomial, since it belongs to
  \(\syzmod{\nodule}{\F}\) and \(\gb\) is reduced.
\end{proof}

In the context of \cref{algo:dac}, \cref{lem:pade:degree} allows us to truncate
the product \(\gbb_2\gbb_1\) while preserving a \(\ord\)-Gr\"obner basis.

\begin{corollary}
  \label{cor:pade:truncate_gb}
  Let \(\nodule\) be as in \cref{eqn:pade_module}, let \(\F \in
  \pmatRing{\rdim}{\edim}\), let \(\ord\) be a monomial order on
  \(\polMod{\rdim}\), and let \(\gb \in \pmatRing{\gbdim}{\rdim}\) be a minimal
  \(\ord\)-Gr\"obner basis of \(\syzmod{\nodule}{\F}\). If \(\gb\) is modified
  by truncating each of its polynomials modulo \(\genBy{X_1^{\prc_1+1}, \ldots,
  X_\nvar^{\prc_\nvar+1}}\), then \(\gb\) is still a minimal \(\ord\)-Gr\"obner
  basis of \(\syzmod{\nodule}{\F}\). 
\end{corollary}
\begin{proof}
  On the first hand, this modification of \(\gb\) does not affect the
  \(\ord\)-leading terms since they all have \(X_i\)-degree less than
  \(\prc_i+1\) according to \cref{lem:pade:degree}, hence after modification we
  still have \(\genBy{\lm{\ord}{\gb}} =
  \genBy{\lm{\ord}{\syzmod{\nodule}{\F}}}\). On the other hand, after this
  modification we also have \(\genBy{\gb} \subseteq \syzmod{\nodule}{\F}\)
  since we started from a basis of \(\syzmod{\nodule}{\F}\) and added to each
  of its elements some multiples of \(\genBy{X_1^{\prc_1+1}, \ldots,
  X_\nvar^{\prc_\nvar+1}}\), which are contained in \(\syzmod{\nodule}{\F}\).
  Then \cite[Lem.\,15.5]{Eisenbud95} yields \(\genBy{\gb} =
  \syzmod{\nodule}{\F}\), hence the conclusion.
\end{proof}

\begin{algorithm}
  \caption{\algoname{Pad\'e}$(\prc_1,\ldots,\prc_\nvar,\G,\ord,\LMinp)$}
  \label{algo:pade}
  \begin{algorithmic}[1]
    \Require{
      \Statex \textbullet~ integers \(\prc_1, \ldots, \prc_\nvar \in \ZZp\),
      \Statex \textbullet~ a matrix \(\G\) in \(\pmatRing{\gbdim}{\edim}\),
      \Statex \textbullet~ a monomial order \(\ord\) on \(\ring^\rdim\),
      \Statex \textbullet~ a list \(\LMinp = (\mmu_1, \ldots, \mmu_\gbdim)\) of elements of \(\monom{\ring^\rdim}\).
    }
    \Ensure{
      \Statex \textbullet~ a matrix \(\gbb\) in \(\pmatRing{\gbbdim}{\rdim}\) for some \(\gbbdim \ge 0\),
      \Statex \textbullet~ a list \(\LM\) of \(\gbbdim\) elements of \(\monom{\ring^\rdim}\).
    }

    \If{\(\prc_1=\cdots=\prc_r=1\)}
      \State \(\gbb \in \pmatRing{\gbdim}{\gbdim} \assign \ident{\gbdim}\);
            \(\H \assign \G \bmod X_1,\ldots,X_\nvar\);
            \(\LM \assign \LMinp\)
      \For{$i=1,\ldots,\edim$}
        \State \(\lf \assign\) linear functional \(\polMod{\edim} \to \field\) defined by \(\lf(\f) = f_i(\boldsymbol{0})\)
        \State \((\gbb_i,\LM) \assign \Call{Syzygy\_BaseCase}{\lf,\H,\ord,\LM}\)
        \State \(\gbb \assign \gbb_i \gbb \bmod X_1^2,\ldots,X_\nvar^2\)
        \label{algo:pade:trunc1} 
        \State \(\H \assign \gbb_i \H \bmod {X_1,\ldots,X_\nvar}\)
        \label{algo:pade:trunc2}
    \EndFor
        \State \Return \((\gbb,\LM)\)
    \EndIf
    \State \(j \assign \max \{i \in \{1,\ldots,\nvar\} \mid \prc_i > 1\}\)
    \State \((\gbb_1,\LM_1) \assign \Call{Pad\'e}{\prc_1,\ldots,\prc_{j-1},\lfloor \prc_{j}/2 \rfloor,1,\ldots,1,\G,\ord,\LMinp}\)
    \State \(\G_2 \assign\) \(X_j^{-\lfloor \prc_j/2 \rfloor}(\gbb_1 \G \mod {X_1^{\prc_1},\ldots,X_{j}^{\prc_{j}},X_{j+1},\ldots,X_\nvar})\)
    \label{algo:pade:division}
    \State \((\gbb_2,\LM_2) \assign \Call{Pad\'e}{\prc_1,\ldots,\prc_{j-1},\lceil \prc_j/2 \rceil,1,\ldots,1,\G_2,\ord,\LM_1}\)
    \label{algo:pade:secondpart}
    \State \(\gbb \assign \gbb_2\gbb_1 \bmod {X_1^{\prc_1+1},\ldots,X_\nvar^{\prc_\nvar+1}}\)
    \label{algo:pade:lastproduct}
    \State \Return{$(\gbb, \LM_2)$}
  \end{algorithmic}
\end{algorithm}

Then, the divide and conquer approach can be refined as described in
\cref{algo:pade}. The correctness of this algorithm can be shown by following
the proof of \cref{thm:algo_general:dac} and with the following considerations.
By induction hypothesis, \(\gbb_1\) is such that each component of the rows of
\(\gbb_1\G\) is an element of
\[
\genBy{X_1^{\prc_1},\ldots,X_{j-1}^{\prc_{j-1}},X_j^{\lfloor\prc_j/2\rfloor},X_{j+1},\ldots,X_\nvar},
\]
hence its truncation modulo 
\[
\genBy{X_1^{\prc_1},\ldots, X_{j}^{\prc_{j}},X_{j+1},\ldots,X_\nvar}
\] 
is an \(\ring\)-multiple of \(X_{j}^{\lfloor \prc_j/2 \rfloor}\). It follows that on
\cref{algo:pade:division}, \(\G_2\) is well defined. Moreover, for \(\p \in
\polMod{\rdim}\) the next equations are equivalent:
\[
\begin{array}{rl}
  \p \gbb_1\G                                   = 0 & { }\bmod X_1^{\prc_{1}}, \ldots, X_{j-1}^{\prc_{j-1}}, X_{j}^{\prc_j} \\[0.2cm]
  \p\G_2 = \p X_j^{-\lfloor \prc_{j}/2 \rfloor} \gbb_1\G = 0 & { }\bmod X_1^{\prc_{1}}, \ldots, X_{j-1}^{\prc_{j-1}}, X_{j}^{\lceil \prc_{j}/2 \rceil}
\end{array}
\]
This justifies the division by \(X_j^{\lfloor \prc_{j}/2 \rfloor}\) at
\cref{algo:pade:division} and the fact that the second call is done with
\(\lceil \prc_j/2 \rceil\) instead of \(\prc_j\) at
\cref{algo:pade:secondpart}. 

For the complexity analysis, we use \cref{lem:pade:degree} to give a bound on
the size of the computed Gr\"obner bases, which differs from the general bound
in \cref{lem:general:size_gb}.

\begin{corollary}[of \cref{lem:pade:degree}]
  \label{cor:pade:size_gb}
  Let \(\nodule\) be as in \cref{eqn:pade_module}, let \(\F \in
  \pmatRing{\rdim}{\edim}\), let \(\ord\) be a monomial order on
  \(\polMod{\rdim}\), and let \(\gb \in \pmatRing{\gbdim}{\rdim}\) be a minimal
  \(\ord\)-Gr\"obner basis of \(\syzmod{\nodule}{\F}\). Then, \[\gbdim \le
  \rdim \prc_1\cdots\prc_\nvar / (\textstyle\max_{1 \le i \le \nvar}
  \prc_i).\]
\end{corollary}
\begin{proof}
  Let \(\LM = \lm{\ord}{\gb} \in \pmatRing{\gbdim}{\rdim}\) and let
  \(\bar\imath\) be such that \(\prc_{\bar\imath} = \max_{1 \le i \le \nvar}
  \prc_i\). It is enough to prove that \(\LM\) has at most
  \(\prc_1\cdots\prc_\nvar / \prc_{\bar\imath}\) rows of the form \(\mu
  \evec{j}\) for each \(j \in \{1,\ldots,\rdim\}\); by \cref{lem:pade:degree},
  the monomial \(\mu \in \monom{\ring}\) has \(X_i\)-degree at most \(\prc_i\)
  for \(1\le i\le \nvar\). Now, for each monomial \(\nu = X_1^{e_1} \cdots
  X_{\bar\imath-1}^{e_{\bar\imath-1}} X_{\bar\imath+1}^{e_{\bar\imath+1}}
  \cdots X_\nvar^{e_\nvar}\) with \(e_i \le \prc_i\) for all \(i \neq
  \bar\imath\), there is at most one row \(\mu \evec{j}\) in \(\LM\) such that
  \(\mu = \nu X_{\bar\imath}^e\) for some \(e \ge 0\): otherwise, one of two
  such rows would divide the other, which would contradict the minimality of
  \(\gb\). The number of such monomials \(\nu\) is precisely
  \(\prc_1\cdots\prc_\nvar / \prc_{\bar\imath}\).
\end{proof}

Here we have \(\vsdim = \edim \prc_1 \cdots \prc_\nvar\), hence the above bound
on the cardinality of minimal \(\ord\)-Gr\"obner bases refines the bound in
\cref{lem:general:size_gb} as soon as \(\rdim \le \edim (\nvar-1)(\max_{1 \le i
  \le \nvar} \prc_i)\).

\begin{proposition}
  \label{prop:bivariate_pade}
  For \(\ring = \field[X,Y]\), let
  \[
    \nodule = \genBy{X^d,Y^e} \times \cdots \times
    \genBy{X^d,Y^e} \subset \polMod{\edim},
  \]
  let \(\F \in \pmatRing{\rdim}{\edim}\) with \(\deg_X(\F) < d\) and
  \(\deg_Y(\F) < e\), and let \(\ord\) be a monomial order on
  \(\polMod{\rdim}\). \cref{algo:pade} computes a minimal \(\ord\)-Gr\"obner
  basis of \(\syzmod{\nodule}{\F}\) using \(\softO{(\maxsz^{\expmm-1} +
  \maxsz\edim)(\maxsz + \edim) d e}\) operations in \(\field\),
  where \(\maxsz = \rdim \min(d,e)\).
\end{proposition}
\begin{proof}
  According to \cref{cor:pade:size_gb}, the number of rows of the matrices
  \(\gbb\) computed in \cref{algo:pade} is at most \(\maxsz = \rdim
  \min(d,e)\). It follows that all matrices \(\gbb_i,\gbb_1,\gbb_2,\gbb\) in
  the algorithm have at most \(\maxsz\) rows and at most \(\maxsz\) columns,
  and that the matrices \(\G,\H,\G_1,\G_2\) have at most \(\maxsz\) rows and
  exactly \(\edim\) columns. Besides, by Kronecker substitution \cite[Chap.\,1
  Sec.\,8]{BiniPan1994}, multiplying two bivariate matrices of dimensions
  \(\maxsz \times \maxsz\) (resp.~\(\maxsz\times \edim\)) and bidegree at most
  \((d,e)\) costs \(\softO{\maxsz^\expmm d e}\) (resp.~\(\softO{\maxsz^\expmm
  (1+\edim/\maxsz)de}\)) operations in \(\field\).

  Let \(\comp(\rdim,\edim,d,e)\) denote the number of field operations used by
  \cref{algo:pade}; we have \(\comp(\rdim,\edim,d,e) \le
  \comp(\maxsz,\edim,d,e)\). First, for \(e > 1\), \(\comp(\maxsz,\edim,d,e)\)
  is bounded by \(\comp(\maxsz,\edim,d,\lfloor e/2 \rfloor) +
  \comp(\maxsz,\edim,d,\lceil e/2 \rceil) +
  \softO{\maxsz^\expmm(1+\edim/\maxsz)de}\). Indeed, there are two recursive
  calls with parameters \((d,\lfloor e/2 \rfloor)\) and \((d,\lceil e/2
  \rceil)\), and two matrix products \(\gbb_1\G\) and \(\gbb_2\gbb_1\) to
  perform; as noted above, the latter products cost
  \(\softO{\maxsz^\expmm(1+\edim/\maxsz)de}\) operations in \(\field\).  The
  same analysis for \(d > 1\) and \(e=1\) shows that
  \(\comp(\maxsz,\edim,d,1)\) is bounded by \(\comp(\maxsz,\edim,\lfloor d/2
  \rfloor,1) + \comp(\maxsz,\edim,\lceil d/2 \rceil,1) +
  \softO{\maxsz^\expmm(1+\edim/\maxsz)d}\).

  Finally, for \(d=e=1\), we show that \(\comp(\maxsz,\edim,1,1) \in
  \bigO{\maxsz (\maxsz + \edim) \edim}\). In this case, there are \(\edim\)
  iterations of the loop. Each of them makes one call to
  \algoname{Syzygy\_BaseCase}, which uses \(\bigO{\maxsz}\) field operations
  for computing the \(\lambda_i\)'s at \cref{algo:basecase:step:lambdas}; note
  that the \(\alpha_j\)'s are zero in the present context where the linear functional
  \(\lf\) corresponds to the constant coefficient. The computed basis
  \(\gbb_i\) has a single nontrivial column (it has the form in
  \cref{eqn:egb_pruned}), so that computing \(\gbb_i\gbb\bmod
  \genBy{X_1^2,\ldots,X_\nvar^2}\) (resp.~\(\gbb_i \H \bmod
  \genBy{X_1,\ldots,X_\nvar}\)) can be done naively at a cost of
  \(\bigO{\maxsz^2}\) (resp.~\(\bigO{\maxsz(\maxsz+\edim)}\)) operations in
  \(\field\).

  Based on the previous inequalities, unrolling the recursion by following the
  divide-and-conquer scheme leads to the announced complexity bound.
\end{proof}

% Fakesection --> fold acknowledgement / biblio (for vimtex plugin, do not remove this line, thanks)
\begin{acks}
    {\it Acknowledgements.}
    The first author acknowledges support from the Fondation Math\'ematique Jacques Hadamard through the Programme PGMO, project number 2018-0061H.
\end{acks}

\bibliographystyle{ACM-Reference-Format}
% below copied from bbl

\end{document}